\documentclass[12pt]{iopart}

\usepackage{iopams}
\usepackage{amsthm, latexsym, setstack, setstack3, amsfonts, amssymb}

\input{Qcircuit}

\newcommand{\ov}{\overline}

\newcommand{\mc}{\mathcal}
\newcommand{\mbb}{\mathbb}

\newcommand{\p}{\prime}
\newcommand{\ra}{\rangle}
\newcommand{\la}{\langle}
\newcommand{\al}{\alpha}

\newcommand{\zd}{{\mbb{Z}_D}}

\newtheorem{theorem}{Theorem}

\newtheorem{proposition}[theorem]{Proposition}
\newtheorem{corollary}[theorem]{Corollary}
\newtheorem{lemma}[theorem]{Lemma}

\newtheorem{theoremcase}{Case}

\begin{document}

\title[An Ideal Characterization of the Clifford Operators]{An Ideal Characterization of the Clifford Operators}
\author{J.~M.~Farinholt}
\address{Electromagnetic and Sensor Systems Department, \\Naval Surface Warfare Ctr, Dahlgren Division, \\Dahlgren, VA 22448}
\ead{Jacob.Farinholt@navy.mil}

\begin{abstract}
The Clifford operators are an important and well-studied subset of quantum operations, in both the qubit and higher-dimensional qudit cases. While there are many ways to characterize this set, this paper aims to provide an ideal characterization, in the sense that it has the same characterization in every finite dimension, is characterized by a minimal set of gates, is constructive, and does not make any assumptions about non-Clifford operations or resources (such as the use of ancillas or the ability to make measurements). While most characterizations satisfy some of these properties, this appears to be the first characterization satisfying all of the above. As an application, we use these results to briefly analyze characterizations of Clifford embeddings, that is, the action of logical Clifford operations acting on qunits embedded in higher-dimensional qudits, inside the qudit Clifford framework.

\end{abstract}

\pacs{03.67.Hk, 03.67.Lx, 03.67.Ac}
\ams{20D45, 81P45, 81P68}

\maketitle
\section{Introduction}

As quantum computers come closer to fruition, an issue of concern is the need for a minimal set of logical quantum gates that can be used to efficiently and reliably implement any universal operation on a quantum state. While the Clifford group is far from universal for quantum computing, understanding this group, especially in higher-dimensional \emph{qudit} systems, is of interest as it provides a rich structure from which novel algorithms may be developed with applications in both finite-dimensional and continuous-variable quantum systems. As many quantum systems naturally exhibit more than two dimensions (e.g. a harmonic oscillator), it is reasonable to investigate methods of manipulating these higher-dimensional states for the purposes of quantum computing. Higher dimensional quantum systems and corresponding Clifford groups have been studied in great detail (see \cite{Knill96-LANL_Report-Arxiv} and \cite{Knill96-LANL_Report2}, which particularly focus on their applications to quantum error correction). In particular, generalizations of the Pauli group for $d$-dimensional systems and associated Clifford group have received additional focus in relation to many applications, including nonbinary stabilizer codes, stabilizer states, graph states, and SIC-POVMs (see, for example, \cite{HostDehMoor05, Gottesman98-arXiv, AshikminKnill01-Nonbinary, Abbpleby09-arXiv, Appleby05-JMathPhys, GrasslRottelerBeth03, NielsenBremnerDoddChilds02-PhysRevA, ApplebyEtAl_QIC12} and references therein). A sufficient set of gates to generate the full Clifford group is known in many higher dimensions. In particular, Gottesman \cite{Gottesman98-arXiv} provided a method of implementing Clifford operators using the SUM gate, ancillas, and measurement operations in quantum systems of any odd prime dimension. In a more general mathematical setting, the Clifford group has been studied as the Lie group associated with the symplectic Clifford algebra (see for example \cite{Crum75, Crum77, Crum90a, oziewicz1992parallel, brackx2001clifford, Dixon81} and references therein).

This paper is motivated by a desire to completely (up to global phase) characterize the Clifford operations in as general and simple a form as possible. In particular, we would first like the characterization to be \emph{universal}, that is, the characterization exists in every finite dimension, so that the study of Clifford operations need not be specific to a particular dimension. Secondly, we would like the characterization to be \emph{physical}, that is, we want to be able to relate this mathematical characterization to realizable quantum transformations, namely, by characterizing the operations with a finite set of gates. Thirdly, we want this characterization to be \emph{minimal}, in the sense that the characterizing set is both necessary and sufficient. Here sufficiency means that any Clifford operator can be characterized by elements of this set, while necessity means that no proper subset of this set can completely characterize every Clifford operator. Lastly, we would like this characterization to be \emph{constructable}; that is, in addition to the knowledge that an arbitrary Clifford operator can be characterized with a minimal set of gates, we would like to know specifically how.

Note, in particular, that such a characterization of the Clifford operators can be considered closed, in the sense that it makes no a priori assumptions on the availability of external resources (e.g. ancillas) or non-Clifford operations (e.g. projective measurements). Such a characterizing set will be called a Clifford \emph{basis}, as it will, indeed, be a generating set for the Clifford group acting via conjugation on the corresponding Pauli group, modulo global phase. 

\section{Qudit Pauli Group\label{Sec:QuditPauliGroup}}

No discussion of Clifford operations is complete without first an overview of the Pauli group. In a $d$-dimensional complex Hilbert space $\mc{H}_d$, a pure state $|\psi\ra$ is generally described as a norm one complex linear sum over the standard \emph{computational} basis. This basis is a collection of $d$ orthonormal states labeled $|0\rangle, |1\rangle, \dots, |d-1\rangle$. Then we have $|\psi \rangle = \al_0|0\rangle + \al_1|1\rangle + \dots + \al_{d-1}|d-1\rangle$, subject to 
\begin{equation}
\sum_{i=0}^{d-1} \left| \al_i\right|^2 = 1.
\end{equation}
In the case of qubits, the Pauli group is a subgroup of the 2-dimensional unitary group generated by $\{iI, \ X, \ Z\}$, where $I$ is the identity matrix, $X = |1\ra\la0| + |0\ra\la1|$, and $Z = |0\ra\la0| - |1\ra\la1|$. While there are many ways to generalize the Pauli group to a $d$-dimensional Hilbert space \cite{Knill96-LANL_Report-Arxiv, Knill96-LANL_Report2}, the simplest and most commonly used is the natural generalization of the qubit Pauli group to $d$-dimensional systems, $\mc{P}_d$, which we now describe. Let $\omega$ be a primitive $d^{th}$ root of unity, i.e. $\omega = \exp(2\pi \text{i}/d)$. We now define the operators
\begin{equation}\label{Eq:XZ-Operators}
X = \sum_{x \in \mbb{Z}_d} |x + 1\rangle \langle x|, \ \ \ 
Z = \sum_{z \in \mbb{Z}_d} \omega^{z}|z\rangle \langle z|,
\end{equation}
with addition defined over the group $\mbb{Z}_d$ on $d$ elements. $X$ and $Z$ each have order $d$. Observe that $(XZ)^r = \omega^{r(r-1)/2}X^rZ^r$. It follows that when $d$ is odd, $XZ$ also has order $d$; however, when $d$ is even, $XZ$ will have order $2d$, contributing additional roots of unity. Thus, we use the notation $\widehat{\omega}$ to denote a primitive $D^{th}$ root of unity, where
\begin{equation}\label{Eq:D}
D = \left\{\begin{matrix} d & \text{ if d is odd,}\\ 2d & \text{ if d is even.} \end{matrix} \right.
\end{equation}
The single-qudit Pauli group $\mc{P}_d$ is defined as the collection of operators $\widehat{\omega}^rX^aZ^b$, where $r \in \mbb{Z}_D$, and $a, b \in \mbb{Z}_d$.

Let $X^a Z^b$ and $X^{a^\p}Z^{b^\p}$ be two operators in $\mc{P}_d$. It is straightforward to see that they have the following commutation relation: 
\begin{equation}\label{Eq:commute}
(X^a Z^b)(X^{a^{\p}}Z^{b^{\p}}) = \omega^{ab^{\p} - ba^{\p}}(X^{a^{\p}}Z^{b^{\p}})(X^a Z^b).
\end{equation}

The Pauli group on an $n$-qudit system is defined as the $n$-fold tensor product of $\mc{P}_d$, denoted $\mc{P}_d^{(n)}$. Ignoring global phase, a typical operator in $\mc{P}_d^{(n)}$ has the form $X^{\textbf{a}}Z^{\textbf{b}} := X^{a_1}Z^{b_1} \otimes X^{a_2}Z^{b_2} \otimes \cdots \otimes X^{a_n}Z^{b_n}$, where $\textbf{a} = (a_1, a_2, \dots, a_n)$ and $\textbf{b} = (b_1, b_2, \dots, b_n)$. The commutation relation of two operators $X^{\textbf{a}}Z^{\textbf{b}}$ and $X^{\textbf{a}^\p}Z^{\textbf{b}^\p}$ in $\mc{P}_d^{(n)}$ is given by
\begin{equation}\label{Eq:commute2}
(X^{\textbf{a}}Z^{\textbf{b}})(X^{\textbf{a}^\p}Z^{\textbf{b}^\p}) = \omega^{(\sum_{i=1}^n a_i b_i^\p - a_i^\p b_i)}(X^{\textbf{a}^\p}Z^{\textbf{b}^\p})(X^{\textbf{a}}Z^{\textbf{b}}).
\end{equation}
This is the natural generalization of \eref{Eq:commute} to the $n$-qudit case.

This relationship gives rise to a classical representation of the Pauli group. Observe that the center of $\mc{P}_d^{(n)}$ is given by $C(\mc{P}_d^{(n)}) = \{\widehat{\omega}^c I \ | \ c \in \mbb{Z}_D \}$, where $I$ is the $n$-fold tensor product of identity.  Since these elements describe the global phase actions on states, we need only consider elements of the quotient group $\mc{P}_d^{(n)} / C(\mc{P}_d^{(n)})$. Elements of this group are equivalence classes containing an operator $X^{\textbf{a}} Z^{\textbf{b}}$ in $\mc{P}_d^{(n)}$ along with all of its complex scalar multiples in $\mc{P}_d^{(n)}$. With a slight abuse of notation, we will label each equivalence class $\{ \widehat{\omega}^c X^{\textbf{a}} Z^{\textbf{b}} \ | \ c \in \mbb{Z}_D \}$ with the scalar-free element $X^{\textbf{a}} Z^{\textbf{b}}$.  While any two elements $X^{\textbf{a}} Z^{\textbf{b}}$ and $X^{\textbf{a}^\p}Z^{\textbf{b}^\p}$ in the quotient group $\mc{P}_d^{(n)} / C(\mc{P}_d^{(n)})$ will always commute, we can still determine whether two elements from their respective preimages in $\mc{P}_d^{(n)}$ will commute by the commutation relation \eref{Eq:commute2}.

The group $\mc{P}_d^{(n)}/C(\mc{P}_d^{(n)})$ is group-isomorphic to the $2n$-dimensional commutative ring module $M_{\mc{R}} = \mbb{Z}_d \times \mbb{Z}_d \times \cdots \times \mbb{Z}_d$ via the map $X^{\textbf{a}}Z^{\textbf{b}} \mapsto (\textbf{a} , \textbf{b})^T = (a_1, a_2, \dots, a_n, b_1, b_2, \dots, b_n)^T$, where multiplication in $\mc{P}_d^{(n)}/C(\mc{P}_d^{(n)})$ becomes addition in $M_{\mc{R}}$. The additional scalar multiplication in the module arises via $(X^{\textbf{a}}Z^{\textbf{b}})^r \mapsto r(\textbf{a} , \textbf{b})^T$. Here we must once again discuss subtle differences between the even and odd case. In the odd case, the ring multiplication is over the integers modulo $d$, $\mbb{Z}_d$. In the even case, however, the ring multiplication is over the integers modulo $2d$, $\mbb{Z}_{2d}$, for reasons discussed earlier. Thus, we write $\mc{R} = \mbb{Z}_D$, where $D$ is defined in \eref{Eq:D}.

By \eref{Eq:commute2}, we preserve the commutative properties of elements in $\mc{P}_d^{(n)}$ by imposing a \emph{symplectic inner product} (SIP) $*$ on the module, where
\begin{equation}\label{Eq:SIP}
(\textbf{a} , \textbf{b})^T * (\textbf{a}^\p , \textbf{b}^\p)^T = \sum_{i=1}^n a_i b_i^\p - a_i^\p b_i \pmod{d}.
\end{equation}
The SIP is a non-degenerate skew-symmetric bilinear form (i.e. a symplectic form) over the ring module, and hence we call the $2n$-dimensional $\mbb{Z}_D$-module with SIP a \emph{symplectic module}. This symplectic module is the \emph{classical representation} of the Pauli group $\mc{P}_d^{(n)}$.

Throughout the rest of this paper, we will omit the superscript $(n)$ and simply write $\mc{P}_d$ to denote the $n$-qudit Pauli group, use $M_{\zd}$ to denote the corresponding symplectic module, and state explicitly when referring to the single-qudit case.

\section{Clifford Operators and Symplectic Form\label{Sec:SymplecticForm}}
The qudit \emph{Clifford group} $\mc{C}$ is defined as the collection of unitary operators that map the Pauli group to itself under conjugation, that is, it is the unitary \emph{normalizer} of the Pauli group. Since the normalizer elements act as automorphisms of the Pauli group, it follows that each Clifford operator has a \emph{classical representation} as a linear operator over $\mbb{Z}_D$ acting on $M_{\mbb{Z}_D}$. While the subset of linear operators over $\mbb{Z}_D$ acting on $M_{\mbb{Z}_D}$ corresponding to Clifford transformations is generally a proper subset, we observe that all transformations of Pauli operators via unitary conjugation preserve the symplectic form. That is, suppose $Q \in U(d)$ is a unitary operator and $X^{\textbf{a}}Z^{\textbf{b}}$, $X^{\textbf{a}^\p}Z^{\textbf{b}^\p} \in \mc{P}_d$. Then
\begin{eqnarray}\label{Eq:CliffordPreservesSIP}
 \ &(Q(X^{\textbf{a}}Z^{\textbf{b}})Q^\dag) (Q (X^{\textbf{a}^\p}Z^{\textbf{b}^\p}) Q^\dag) \\
 &= Q(X^{\textbf{a}}Z^{\textbf{b}})(X^{\textbf{a}^\p}Z^{\textbf{b}^\p}) Q^\dag \\
 &= \omega^{\left((\textbf{a} , \textbf{b})^T * (\textbf{a}^\p , \textbf{b}^\p)^T\right)}Q(X^{\textbf{a}^\p}Z^{\textbf{b}^\p})(X^{\textbf{a}}Z^{\textbf{b}})Q^\dag\\
 &= \omega^{\left((\textbf{a} , \textbf{b})^T * (\textbf{a}^\p , \textbf{b}^\p)^T\right)}(Q(X^{\textbf{a}^\p}Z^{\textbf{b}^\p})Q^\dag) (Q(X^{\textbf{a}}Z^{\textbf{b}})Q^\dag).
\end{eqnarray}
It follows that the Clifford operators must be classically represented as $2n \times 2n$ matrices over $\mbb{Z}_D$ acting on $M_\zd$ that preserve the symplectic form. Such matrices $N$ are called \emph{symplectic matrices}, and they satisfy $N^TSN = S$, where
\begin{equation}\label{Eq:S}
S = \begin{bmatrix} 0&I_n\\ -I_n&0\end{bmatrix}
\end{equation}
is determined by the symplectic form, namely,
\begin{equation}
(\textbf{a}, \textbf{b})^T * (\textbf{a}^\p, \textbf{b}^\p)^T = (\textbf{a}, \textbf{b}) S (\textbf{a}^\p, \textbf{b}^\p)^T.
\end{equation}
It is straightforward to see that a $2\times 2$ matrix $N$ with entries in $\mbb{Z}_D$ is symplectic if and only if $\det(N) = 1 \pmod{D}$.

In \cite{Appleby05-JMathPhys}, Appleby shows that preserving the symplectic form is also a sufficient condition for the description of Clifford operators in the single-qudit case, and Hostens \emph{et. al.} \cite{HostDehMoor05} proved the sufficiency in the multi-qudit case, using a slightly different classical construction and large classes of gates. In neither of these references is a basis developed for the Clifford group. The proofs for the basis developed in this paper are constructive, giving rise to an algorithm to implement any Clifford transformation using only compositions of three basis gates.

\section{Building the Clifford Group from a Generating Set\label{Sec:BuildClifford}}
Gottesman \cite{Gottesman98-arXiv} presented a sufficient set of gates to generate the Clifford group when the Hilbert space has odd prime dimension $d$. In what follows, we constructively prove that a Clifford basis can be built in arbitrary dimension from a smaller subset of those gates. We do this by first constructing the single-qudit Clifford group by explicitly building any arbitrary $2 \times 2$ symplectic operator using sequences of two gates, namely the discrete Quantum Fourier Transform (QFT) and Phase-shift gates. We then show that the addition of the two-qudit SUM gate acting on pairs of qudits generates the entire $n$-qudit Clifford group.

Throughout this paper, we rely heavily on the classical representation of Clifford operators as elements of a symplectic group to reveal information about the actions of these operators on quantum states. Because any unitary Clifford operator can be uniquely (up to global phase) represented by a symplectic matrix, we take measures in this paper to specifically identify which representation is being referred to. When discussing Clifford gates and Clifford operators outside of a particular mathematical representation, we identify them by their names (i.e. SUM, QFT, Phase-shift, etc.). When specifically referring to the classical representation, these operators are identified by a single capital letter. When a line is placed above one of these letters, we are referring to (up to global phase) the unitary representation of the corresponding Clifford operator.

\subsection{The Single-Qudit Clifford Group}

The QFT and Phase-shift gates can be described (up to global phase) by their action on $X$ and $Z$ under conjugation. The QFT maps
\begin{eqnarray}\label{Eq:QFT-map}
X &\mapsto Z \\
Z &\mapsto X^{-1} \label{Eq:QFT-map2}
\end{eqnarray}
and the Phase-shift maps
\begin{eqnarray}\label{Eq:Phase-map}
X &\mapsto XZ\\
Z & \mapsto Z.\label{Eq:Phase-map2}
\end{eqnarray}
Thus, we can classically represent the QFT gate with the symplectic matrix
\begin{equation}
R = \begin{bmatrix}0&-1\\1&0\end{bmatrix}
\end{equation}
and the Phase-shift gate with the symplectic matrix
\begin{equation}
P = \begin{bmatrix}1&0\\1&1\end{bmatrix},
\end{equation}
where the matrix entries are taken modulo $D$.
\begin{proposition}\label{Prop:LC-group}
The Phase-shift and QFT gates are a necessary and sufficient set of gates to generate (up to global phase) the entire single-qudit Clifford group in any finite dimension.
\end{proposition}
Before proving this, we must first discuss what is known as the Pauli-Euclid-Gottesman (PEG) Lemma \cite{Gottesman98-arXiv, NielsenBremnerDoddChilds02-PhysRevA}.

\begin{lemma}[PEG Lemma \cite{Gottesman98-arXiv, NielsenBremnerDoddChilds02-PhysRevA}]\label{Lem:PEG}
For any dimension $d$ and for integers $0 \leq j, k \leq d-1$, there exists a Clifford operator mapping $X^jZ^k$ to $Z^{\gcd(j,k)}$.
\end{lemma}
This result, first alluded to in \cite{Gottesman98-arXiv} and formally proved in \cite{NielsenBremnerDoddChilds02-PhysRevA}, is accomplished by using the operators $P$ and $R$ to implement Euclid's factorization algorithm, mapping the symplectic vector $( a \ , \ b)^T$ to $(0 \ , \ \gcd(a,b))^T$. It is easy to see how this is done by observing that by appropriate applications of $P$ and $R$, the vector $(a \ , \ b)^T$ can be mapped to $(a \ , \ b+ ma)^T$ or $(a + mb \ , \ b)^T$ for any integer $m$. The appropriate choice of the value $m$ is determined at each step of the Euclidean algorithm. Note that in general, the exponents of $X$ and $Z$ will be in the range $\{0, 1, \dots, d-1\}$. However, $\gcd(0,m)$ is undefined for any nonzero integer $m$. Thus, in order for the algorithm to be defined over all values $0 \leq m \leq d-1$, we define $\gcd(0, m) = \gcd(m, 0) = m$. This definition is appropriate, as $m = k0 + m$ for any $k$, and hence, by Euclid's factorization algorithm, $\gcd(0,m) = \gcd(m,m) = m$.

\begin{proof}[Proof of Proposition \ref{Prop:LC-group}] Since (up to global phase) every single-qudit Clifford operator is uniquely represented by a $2 \times 2$ symplectic matrix, we may prove sufficiency by explicitly generating any arbitrary $2 \times 2$ symplectic matrix $M$ using only these two gates. Suppose the Clifford operator we want to build has a classical representation given by
\begin{equation}
M = \begin{bmatrix}p&q\\r&s\end{bmatrix},
\end{equation}
where the entries are over $\mbb{Z}_D$. Since $M$ is symplectic by assumption, we know that $\det(M) = 1$. Now we consider two distinct cases, dependent on the invertibility of the entries in $M$.
\begin{theoremcase}\label{Case:q-invert}
Suppose $q$ is invertible. Then it is easily verified that $M = P^mRP^qRP^n$, where $m = q^{-1}(s+1)$ and $n = q^{-1}(p+1)$. Furthermore, if any entry in $M$ is invertible, then $M$ can be decomposed in a like manner, with appropriate additional applications of $R$.
\end{theoremcase}

What is meant by ``appropriate additional applications of $R$'' is decribed by the following. Suppose, for example, that $q$ is not invertible, but $r$ is. Then $M^* = RMR$ has $r$ in the top right position, and can be decomposed accordingly. The original $M$ is obtained by observing that $M = RM^*R$. Thus, as long as there is at least one invertible element in a $2 \times 2$ symplectic operator $M$, it can be easily decomposed into a product of $P$ and $R$.

\begin{theoremcase}\label{Case:no-invert}
Suppose no entries in $M$ are invertible. Since $M$ is invertible ($M^{-1} = -SM^TS$, where $S$ is the symplectic form matrix from \eref{Eq:S}), it follows that for any column $M_{.,i}$ in $M$, $\gcd(M_{.,i})$ is invertible. Thus, we can use the algorithm described in the PEG Lemma to map $M$ to some $M^\p$ having an invertible element in one of the columns. We build $M^\p$ as in Case \ref{Case:q-invert} and then use inverse operations to obtain $M$.
\end{theoremcase}
Thus, we conclude that QFT and Phase-shift are a sufficient set of gates to generate any single-qudit Clifford group in any finite dimension. The proof of necessity is straightforward. The matrix $R$ cannot generate the matrix $P$; it follows that the Phase-shift gate cannot be obtained from the QFT. Likewise, the QFT cannot be obtained from the Phase-shift gate, and hence, neither one can generate the entire Clifford group on its own. Thus both are necessary, completing the proof.
\end{proof}

A more detailed overview of how to apply the \emph{PEG Algorithm} described in Case \ref{Case:no-invert} to a particular column of a symplectic matrix is described in \ref{PEG-Algorithm}, along with a simple example. Note that, while in general, $\mc{O}(D \ln{D})$ gates are needed to implement a Clifford operator in a $d$-dimensional Hilbert space, in the case of prime dimension, only a linear number of gates are needed, as only Case \ref{Case:q-invert} applies.

Because qudit Clifford transformations are automorphisms on the qudit Pauli group, it follows that, at a minimum, if a qudit Clifford transformation mapping $Z^i$ to $Z^j$ exists, then $\gcd(i,d) = \gcd(j,d)$. Suppose $\gcd(i,d) = \gcd(j,d) = k$ for some $1 \leq k \leq d$. Then $i = pk$ and $j = qk$ for some $p$ and $q$ both relatively prime to $d$. Since $p$ and $q$ are then elements of $\mbb{Z}_d^*$, the (cyclic) multiplicative group of elements relatively prime to $d$ (modulo d), it follows that $q = p^r \pmod{d}$ for some $1 \leq r \leq d-1$, and hence $j = p^{r-1}i \pmod{d}$. Conversely, if $j = ki \pmod{d}$ for some $k$ relatively prime to $d$, then there exists a Clifford transformation mapping $Z^i$ to $Z^j$ up to some global phase, whose symplectic representation is given by
\begin{equation}\label{Eq:S_k}
S(k) = \begin{bmatrix}k^{-1} & 0\\ 0 & k\end{bmatrix} = RP^{(k^{-1})}RP^kRP^{(k^{-1})}.
\end{equation}
Combining this result with Lemma \ref{Lem:PEG} gives the following corollary.
\begin{corollary}\label{Cor:S_k}
For any dimension $d$ and for integers $0 \leq a_1, b_1, a_2, b_2 \leq d-1$, there exists a Clifford operator mapping $X^{a_1}Z^{b_1}$ to $X^{a_2}Z^{b_2}$ if and only if $\gcd(a_1, b_1) = k\gcd(a_2, b_2)$ for some $k$ relatively prime to $d$.
\end{corollary}

\subsection{Multi-Qudit Clifford Transformations}

Before generalizing the above results to the multi-qudit case, we introduce a multi-qudit gate called the SUM gate, which is the natural generalization of the CNOT gate to arbitrary $d$-dimensional quantum systems. It is a well-known non-local Clifford operator that maps
\begin{eqnarray}\label{Eq:SUMmap}
X \otimes I &\mapsto X\otimes X,\\
I \otimes X & \mapsto I \otimes X,\\
Z \otimes I & \mapsto Z \otimes I,\\
I \otimes Z & \mapsto Z^{-1} \otimes Z,\label{Eq:SUMmap4}
\end{eqnarray}
 and is classically represented by the following symplectic operator:
\begin{equation}\label{Eq:SUMgate}
C = \begin{bmatrix} 1&0&0&0\\
                      1&1&0&0\\
                      0&0&1&-1\\
                      0&0&0&1 \end{bmatrix}.
\end{equation}

In an effort to help the reader see that this is indeed the classical representation of the SUM gate, we recall here that the classical representation of the multi-qudit Pauli operators is a vector that first lists the $X$ exponents, then the $Z$ exponents in the tensor product. That is, $X \otimes X \mapsto (1, 1, 0, 0)^T$, $X \otimes Z \mapsto (1, 0, 0, 1)^T$, etc. Thus, for example, since the Clifford operator maps $X \otimes I \mapsto X \otimes X$, the classical representation must be a symplectic operator that maps $(1, 0, 0, 0)^T \mapsto (1, 1, 0, 0)^T$, etc.

Note that the transpose, $C^T$, corresponds to the SUM gate in which the control and target qudit are switched. For purposes of notational clarity, the SUM gate in which qudit $i$ is the control and qudit $j$ is the target will be denoted by $C_{[i,j]}$. Thus $C_{[j,i]} = C_{[i,j]}^T$. In \ref{Appendix:PEG-SUM}, we show that, by using sequences of $C_{[1,2]}$ and $C_{[2,1]}$, we can use an approach similar to that described in Lemma \ref{Lem:PEG} to map $(0,0,a,b)^T$ to either $(0,0, \gcd(a,b), 0)^T$ or $(0,0,0, \gcd(a,b))^T$ for any pair $0 \leq a,b \leq d-1$. Thus, we have the following result.

\begin{lemma}\label{Lem:SUM-adapted_PEG}
For any dimension $d$ and for integers $0 \leq a, b \leq d-1$, there exists a Clifford operator mapping $Z^a \otimes Z^b$ to $I \otimes Z^{\gcd(a, b)}$.
\end{lemma}
Again, this is well-defined given our definition of $\gcd(0,m) = \gcd(m,0) = m$ for all integers $m$. It is straightforward to see that this result generalizes to the $n$-qudit Pauli operators as well. Hence, combining this with Lemma \ref{Lem:PEG} results in the following proposition.

\begin{proposition}\label{Prop:GeneralizedPEG}
For any dimension $d \geq 2$, any positive integer $n$, and integers $0 \leq a_i, b_j \leq d-1$, there exists a Clifford operator mapping $X^{a_1}Z^{b_1} \otimes \cdots \otimes X^{a_n}Z^{b_n}$ to $I^{\otimes n-1} \otimes Z^k$, where $k = \gcd(a_1, a_2, \dots, a_n, b_1, b_2, \dots, b_n)$.
\end{proposition}
In the above proposition, $I^{\otimes n-1}$ denotes the $(n-1)$-fold tensor product of identity. The above result is accomplished by first using the PEG algorithm to map each $X^{a_i}Z^{b_i}$ to $Z^{\gcd(a_i,b_i)}$, and then using the SUM-adapted PEG Algorithm to first map $Z^{\gcd(a_1,b_1)}\otimes Z^{\gcd(a_2,b_2)}$ to $I \otimes Z^{\gcd(a_1, a_2, b_1, b_2)}$, then applying it again to map $Z^{\gcd(a_1, a_2, b_1, b_2)} \otimes Z^{\gcd(a_3, b_3)}$ to $I \otimes Z^{\gcd(a_1, a_2, a_3, b_1, b_2, b_3)}$. Continuing in this fashion produces the final result. We call this algorithm the \emph{Generalized PEG Algorithm}, and it requires $\mc{O}(n)\mc{O}(D \ln(D))$ gates to implement.

Using this result, we can generalize Corollary \ref{Cor:S_k} to the multi-qudit case.

\begin{corollary}\label{Cor:S_k-Generalized}
For any dimension $d$ and for integers $0 \leq a_1, b_1, \dots, a_n, b_n, a_1^\p, b_1^\p, \dots, a_n^\p, b_n^\p \leq d-1$, there exists a Clifford operator mapping $X^{a_1}Z^{b_1}\otimes X^{a_2}Z^{b_2}\otimes \cdots \otimes X^{a_n}Z^{b_n}$ to $X^{a_1^\p}Z^{b_1^\p}\otimes X^{a_2^\p}Z^{b_2^\p}\otimes \cdots \otimes X^{a_n^\p}Z^{b_n^\p}$ if and only if $\gcd(a_1, b_1, \dots, a_n, b_n) = k\gcd(a_1^\p, b_1^\p, \dots, a_n^\p, b_n^\p)$ for some $k$ relatively prime to $d$. Moreover, such an operator can be constructed using $\mc{O}(n)\mc{O}(D \ln(D))$ gates.
\end{corollary}

The above mentioned operator can be implemented using a meet-in-the-middle approach to the generalized PEG algorithm. Namely, use generalized PEG to determine an operator $M$ constructed from a composition of gates that map $ P_1 = X^{a_1}Z^{b_1}\otimes \cdots \otimes X^{a_n}Z^{b_n}$ to $I^{\otimes n-1}\otimes Z^{\gcd(a_1, \dots, a_n, b_1, \dots, b_n)}$. Then we use the algorithm again to determine an operator $N$ constructed from a composition of gates that map $P_2 = X^{a_1^\p}Z^{b_1^\p} \otimes \cdots \otimes Z^{a_n^\p}Z^{b_n^\p}$ to $I^{\otimes n-1} \otimes Z^{\gcd(a_1^\p, \dots, a_n^\p, b_1^\p, \dots, b_n^\p)}$. Since the algorithm is reversible, we use the inverses of these gates to obtain an operator $N^{-1}$. Then since $S(k)$ applied to the last qudit has the effect of mapping $I^{\otimes n-1}\otimes Z^{\gcd(a_1, \dots, a_n, b_1, \dots, b_n)}$ to $I^{\otimes n-1}\otimes Z^{k\gcd(a_1, \dots, a_n, b_1, \dots, b_n)} = I^{\otimes n-1} \otimes Z^{\gcd(a_1^\p, \dots, a_n^\p, b_1^\p, \dots, b_n^\p)}$, it follows that $(N^{-1}S(k)_{[n]}M)P_1(N^{-1}S(k)_{[n]}M)^\dag = P_2$, where we use $S(k)_{[n]}$ to indicate that $S(k)$ is acting on the $n$-th qudit.

\subsection{The Multi-Qudit Clifford Group and Basis}

Before stating the main result, we first list the classical representations of various Clifford operators. The Phase-Shift gate acting on the $i$-th qudit of an $n$-qudit state is classically represented by the symplectic matrix
\begin{equation}\label{Eq:P_i}
P_{[i]} = \begin{bmatrix} I & 0_n\\ E_{i,i} & I \end{bmatrix},
\end{equation}
where $I$ is the $n \times n$ identity matrix, $0_n$ is the $n \times n$ all-zero matrix, and $E_{i,i}$ is a matrix of all zeros, except for a 1 in the $i$-th diagonal entry. Another operator of importance is given by the transpose of the above matrix, constructed by a particular product of QFT and Phase-Shift gates acting on the $i$-th qudit of an $n$-qudit state, namely,
\begin{equation}\label{Eq:P_i^T}
(P_{[i]})^T = R_{[i]}P^{-1}_{[i]}R_{[i]}^3 = \begin{bmatrix} I & E_{i,i}\\ 0_n & I \end{bmatrix}.
\end{equation}
Now suppose $1 \leq i < j \leq n$. Then the SUM gate acting on an $n$-qudit system using $i$ as the control and $j$ as the target qudit is classically represented by the symplectic matrix
\begin{equation}
C_{[i,j]} = \begin{bmatrix} E_{j,i} + I & 0_n\\ 0_n & I - E_{i,j} \end{bmatrix},
\end{equation}
where each $E_{p,q}$ is an $n \times n$ matrix of all-zeros except for a $1$ in the $(p,q)$-th entry. Note that, just as in the two-qudit case, the SUM gate in which the control and target qudits are reversed is given classically by the transpose, $C_{[j,i]} = C_{[i,j]}^T$. We are now ready to state and prove our main result.

\begin{theorem}\label{Thm:CliffordGroup}
In any dimension $d \geq 2$ and for any number $n$ of qudits, a necessary and sufficient set of gates to generate (up to global phase) the $n$-qudit Clifford group (i.e. a Clifford basis) is given by the discrete QFT and Phase-shift gates acting on individual qudits, and the SUM gate acting on pairs of qudits.
\begin{proof}
It is clear that, without the use of ancillas, no one of these gates can be constructed from the other two, proving the necessity. We have already shown how to construct the single-qudit Clifford group using QFT and Phase-shift gates. We will prove the $n$-qudit case by induction. Namely, we will use the QFT, Phase-shift, and SUM gates to map an arbitrary $n$-qudit Clifford operator to another $n$-qudit Clifford operator that acts as identity on the last qudit. Such an operator is equivalent to an $(n-1)$-qudit Clifford operator acting on the first $n-1$ qudits, and hence the results will follow by induction.

Let $N = \begin{bmatrix}J&K\\L&M\end{bmatrix}$ be an arbitrary $2n \times 2n$ symplectic matrix with entries over $\mbb{Z}_D$. Just as in the single qudit case, we can use the generalized PEG algorithm to map $N$ to another symplectic matrix in which the last column has all zeros except for the very bottom entry, given by the greatest common divisor of all of the entries in that column, labelled $k$. Because $N$ is invertible, it follows that $k$ is invertible. Thus, we can apply $S(k^{-1})_{[n]}$ on the left to map $k$ to 1. Note that, by symplecticity, it follows that $J_{n,n}$ is now 1 as well. We now wish to map each of the remaining entries in the bottom row of $N$ to 0.

Let $M_{n,i}$ and $L_{n,i}$ denote the $i$-th entry in the last row of $M$ and $L$, respectively. We map each $M_{n,i}$ for $1 \leq i <n$ to zero by applying $C_{[n,i]}^{M_{n,i}}$ to $N$ on the right. We map $L_{n,n}$ to zero by applying $P_{[n]}^{-L_{n,n}}$ to $N$ on the right. In order to map $L_{n,i}$ to zero for the remaining $1 \leq i < n$, we apply $R_{[i]}P_{[i]}R_{[i]}P_{[i]}R_{[i]}^2$ on the right of $N$. This effectively replaces the $i$-th column of $L$ with that of $M$, which has a zero in the last position. Unfortunately, this also causes each of the values in the bottom row of $M$ to become nonzero, so we have to repeat some steps to map those back to zero.

After these steps we have a symplectic matrix $N^\p$ in which the last column and bottom row are all zeros except for the last entry $N_{2n,2n}^\p$ being a 1. By symplecticity, it follows that the $n$-th column and $n$-th row are all zeros except for a 1 in position $N_{n,n}^\p$.
This corresponds to a Clifford operator that acts as identity on the last qudit. Hence, this corresponds to an $(n-1)$-qudit Clifford operator tensored with the single-qudit identity operator. By induction, the statement of the theorem is obtained.
\end{proof}
\end{theorem}

Note that it is implicit by the above proof that every symplectic matrix corresponds to a Clifford operator. The algorithm described in this proof uses $\mc{O}(n^2) \mc{O}(D \ln(D))$ gates to implement.

As an example of how we can generate multi-qudit Clifford gates from this finite set, we construct the two qudit SWAP gate. This gate performs the operation $|i\ra|j\ra \mapsto |j\ra|i\ra$, where $i,j \in \{0, 1, \dots, d-1\}$. When this gate acts via conjugation on the two-qudit Pauli group, it performs the operation $X^{a_1}Z^{b_1}\otimes X^{a_2}Z^{b_2} \mapsto X^{a_2}Z^{b_2}\otimes X^{a_1}Z^{b_1}$, and hence is classically represented as the symplectic matrix
\begin{equation}\label{Eq:SWAP}
S_{[i,j]} = \begin{bmatrix} 0&1&0&0\\1&0&0&0\\0&0&0&1\\0&0&1&0\end{bmatrix},
\end{equation}
where the $[i,j]$ are to indicate on which pair of qudits the operator is acting. This gate is implemented by performing sequences of SUM and local QFTs. Let $R_{[i]}$ denote the QFT acting on qudit $i$, and $R_{[i,j]}$ denote the QFT acting transversally on qudits $i$ and $j$, so that $R_{[i,j]} = R_{[i]}R_{[j]} = R_{[j,i]}$. Then the SWAP gate can be classically decomposed as
\begin{equation}\label{Eq:SWAP-decomp}
S_{[i,j]} = R_{[j]}R_{[j]}C_{[i,j]}R_{[i,j]}C_{[i,j]}R_{[i,j]}C_{[i,j]}.
\end{equation}

We include a circuit diagram for the implementation of SWAP in Figure \ref{Fig:SWAP}.
Note that, in the qubit case, it is known (see, for example \cite{BellEtAl_Arxiv13} and \cite{NielChuang}) that SWAP can be implemented using CNOT gates alone. This follows from the fact that, in the qubit case (and the qubit case alone), $R^2 = I$, the identity operator, and $R_{[i,j]}C_{[i,j]}R_{[i,j]} = C_{[j,i]}$.

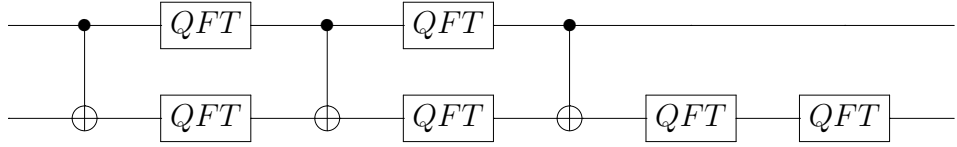
\begin{figure}
\begin{equation*}
 \Qcircuit @R=1.5em {
  & \ctrl{1} & \gate{QFT} & \ctrl{1} & \gate{QFT} & \ctrl{1} & \qw & \qw & \qw \\
  & \targ & \gate{QFT} & \targ & \gate{QFT} & \targ & \gate{QFT} & \gate{QFT} & \qw
  }\end{equation*}
  \caption{Circuit diagram of SWAP implementation using SUM and QFT gates. The QFT gates are labeled accordingly, and the SUM gates are indicated by a vertical line, with the solid dot on the control qudit and a $\oplus$ on the target qudit.\label{Fig:SWAP}}
\end{figure}

In \cite{NielsenBremnerDoddChilds02-PhysRevA}, the unitary representation of the QFT and Phase-shift gates were explicitly defined for any dimension $d$, and we include them here for completeness. The discrete QFT $\ov{R}$ is defined by
\begin{equation}\label{Eq:QFT}
  \ov{R}|j\ra \equiv \sum_{k=0}^{d-1} \omega^{jk}|k\ra
\end{equation}
for each $j \in \{0, 1, \dots, d-1 \}$. The unitary representation of the Phase-shift gate $\ov{P}$ is dependent on whether the dimension $d$ is even or odd. For odd $d$, $\ov{P}$ is defined by
\begin{equation}\label{Eq:Phase-odd}
  \ov{P}|j\ra \equiv \omega^{j(j-1)/2}|j\ra,
\end{equation}
and for even $d$, $\ov{P}$ is defined by
\begin{equation}\label{Eq:Phase-even}
  \ov{P}|j\ra \equiv \omega^{j^2/2}|j\ra
\end{equation}
for each $j \in \{0, 1, \dots, d-1 \}$. Straightforward computations verify that $\ov{R}X\ov{R}^\dagger = Z$, $\ov{R}Z\ov{R}^\dagger = X^{-1}$, satisfying equations \eref{Eq:QFT-map} and \eref{Eq:QFT-map2},  and $\ov{P}X\ov{P}^\dagger = XZ$, $\ov{P}Z\ov{P}^\dagger = Z$ when $d$ is odd, and $\ov{P}X\ov{P}^\dagger = \omega^{1/2}XZ$, $\ov{P}Z\ov{P}^\dagger = Z$ when $d$ is even, so that, up to global phase, equations \eref{Eq:Phase-map} and \eref{Eq:Phase-map2} are satisfied.

In \cite{Gottesman98-arXiv}, the unitary representation of the SUM gate, $\ov{C}$, was given by 
\begin{equation}\label{Eq:SUM-odd}
\ov{C}|i\ra|j\ra \equiv |i\ra|i + j\ra,
\end{equation}
for all $i,j \in \{0, 1, \dots, d-1\}$, with addition defined modulo $d$. This definition works for odd dimensions; however, it must be changed slightly for even $d$. In this case, we define
\begin{equation}\label{Eq:SUM-even}
\ov{C}|i\ra|j\ra \equiv \omega^{(i + j)/2} |i\ra|i + j\ra.
\end{equation}
It is easily verified that, when $d$ is odd, we have
\begin{eqnarray}\label{Eq:unitarySUM-odd}
\ov{C}(X\otimes I) \ov{C}^\dag &= X \otimes X\\
\ov{C}(I\otimes X) \ov{C}^\dag &= I \otimes X\\
\ov{C}(Z\otimes I) \ov{C}^\dag &= Z \otimes I\\
\ov{C}(I\otimes Z) \ov{C}^\dag &= Z^{-1} \otimes Z,
\end{eqnarray}
and when $d$ is even,
\begin{eqnarray}\label{Eq:unitarySUM-even}
\ov{C}(X\otimes I) \ov{C}^\dag &= \omega^{1/2}(X \otimes X)\\
\ov{C}(I\otimes X) \ov{C}^\dag &= \omega^{1/2}(I \otimes X)\\
\ov{C}(Z\otimes I) \ov{C}^\dag &= \omega^{1/2}(Z \otimes I)\\
\ov{C}(I\otimes Z) \ov{C}^\dag &= \omega^{1/2}(Z^{-1} \otimes Z).
\end{eqnarray}
Hence, up to global phase, equations \eref{Eq:SUMmap} through \eref{Eq:SUMmap4} hold in both the even and odd cases. Thus, up to global phase, every unitary Clifford operator is defined for any finite-dimensional Hilbert space.

\section{Clifford Embeddings}\label{Sec:CliffordEmbed}
In \cite{GKP_PRA2001}, Gotessman \emph{et al.} addressed the possibility of embedding an $n$-dimensional qunit into a $d$-dimensional qudit for $d>n$ to protect against small shift errors. Here we review this embedding procedure, which we will refer to as the \emph{GKP procedure} (for Gottesman, Kitaev, and Preskill), and then analyze the constructions of the corresponding logical Clifford operators under these embeddings. In particular, we show that under some asymetric embeddings, many logical qunit Clifford operations cannot be constructed from qudit Clifford operations.

Under the GKP procedure, if $d = r_xr_zn$, then we may embed an $n$-dimensional qunit inside a $d$-dimensional qudit by defining the logical qunit operators $X_L = X^{r_x}$ and $Z_L = Z^{r_z}$. The encoded computational basis for the qunit is given by
\begin{equation}\label{Eq:LogicalBasis}
|j\ra_L = \frac{1}{\sqrt{r_z}} \sum_{i = 0}^{r_z-1} |(j+in)r_x\ra,
\end{equation}
for $j = 0, 1, \dots, n-1$; the Fourier transform of this basis looks the same, except with $r_x$ and $r_z$ interchanged everywhere. This embedding has a corresponding stabilizer group generated by $X^{nr_x}$ and $Z^{nr_z}$, and is protected against all shifts of the form $X^aZ^b$ for $|a| < r_x/2$ and $|b| < r_z/2$.

Under this construction, the logical Clifford basis (and hence the entire logical qunit Clifford group) can be completely characterized by the following transformations. Up to global phase, the logical QFT must map
\begin{eqnarray}\label{Eq:LogicalQFT-map}
X^{r_x} &\mapsto Z^{r_z} \\
Z^{r_z} &\mapsto X^{-r_x} \label{Eq:LogicalQFT-map2}
\end{eqnarray}
and the logical Phase-shift must map
\begin{eqnarray}\label{Eq:LogicalPhase-map}
X^{r_x} &\mapsto X^{r_x}Z^{r_z}\\
Z^{r_z} & \mapsto Z^{r_z},\label{Eq:LogicalPhase-map2}
\end{eqnarray}
while the logical SUM must map
\begin{eqnarray}\label{Eq:LogicalSUMmap}
X^{r_x} \otimes I &\mapsto X^{r_x}\otimes X^{r_x},\\
I \otimes X^{r_x} & \mapsto I \otimes X^{r_x},\\
Z^{r_z} \otimes I & \mapsto Z^{r_z} \otimes I,\\
I \otimes Z^{r_z} & \mapsto Z^{-r_z} \otimes Z^{r_z}.\label{Eq:LogicalSUMmap4}
\end{eqnarray}

Note that $X^aZ^b$ acts as identity on the embedded space whenever $a$ (respectively $b$) is a multiple of $nr_x$ (respectively $nr_z$). As such, we may also allow products of such $X^aZ^b$ on the right hand side of each of these maps (e.g. $X^{r_x} \mapsto Z^{r_z}$ is equivalent to $X^{r_x} \mapsto X^{\alpha_1nr_x}Z^{(1+\alpha_2n)r_z}$ for arbitrary integers $\alpha_1$ and $\alpha_2$).

It is straightforward to see that in the symmetric case (i.e. $r_x = r_z$), we have the nice property that the qudit QFT, Phase-shift and SUM operators in fact act as logical qunit QFT, Phase-shift and SUM operators, respectively. Thus, in any symmetric embedding, all logical Clifford gate and Clifford circuit analysis can be immediately inferred from the corresponding qudit Clifford gates and circuits. Moreover, as long as the errors are small enough relative to $r = r_x=r_z$, the qudit Clifford gates are a fault-tolerant set of logical qunit Clifford gates.

In the asymetric case, however, all of these nice properties are generally not true. For example, consider the case in which $r_x = 3$, $r_z = 4$, and $n= 2$, so that $d = 24$. In this case, the qudit SUM gate applies the logical qubit CNOT transformation. However, the logical Phase-shift and QFT cannot be constructed via qudit Clifford operations. We will show this by contradiction in the case of the logical QFT.

Supposing the logical qubit QFT could be applied by a qudit Clifford operator, then we could classically describe the logical qubit QFT as a $2 \times 2$ symplectic matrix $M$ with entries in $\mbb{Z}_{48}$ such that
\begin{equation}
M\begin{bmatrix}3&0\\0&4\end{bmatrix} = \begin{bmatrix}0&-3\\4&0\end{bmatrix} = \begin{bmatrix}0&21\\4&0\end{bmatrix} \pmod{24}.
\end{equation}
Letting $M = \begin{bmatrix} a&b\\c&d\end{bmatrix}$, we see that, in particular, we must have $4b = 21 \pmod{24}$. Clearly, no such $b$ exists. In fact, we can multiply the desired outcomes by $X^{6\alpha_1}Z^{8\alpha_2}$ for arbitrary integers $\alpha_1$ and $\alpha_2$ (since this is equivalent to multiplication by identity on the embedding) and do the same comparison and likewise determine that no such symplectic matrix exists. That is, there exists no construction of $M$ that satisfies
\begin{equation}
M\begin{bmatrix}3&0\\0&4\end{bmatrix} = \begin{bmatrix}6\alpha_1&21+ 6\alpha_2\\4 + 8\alpha_3&8\alpha_4\end{bmatrix} \pmod{24},
\end{equation}
for arbitrary integers $\alpha_i$. It follows that no qudit Clifford operator exists that performs the logical qubit QFT. Similar arguments can be used to come to the same conclusion for the logical qubit Phase-shift. We say then that this embedding is \emph{non-symplectic}.

Thus, in general, the logical Clifford operators over asymetric embeddings of a qunit into a qudit cannot be analyzed inside the qudit Clifford framework. Note that, under \emph{any} embedding of a qunit into a qudit, the qudit SUM gate will always apply a logical qunit SUM gate. Thus, to determine whether a particular asymetric embedding is symplectic, we need only check against two gates; that is, we need only check whether logical qunit QFT and Phase-shift operations can be constructed from qudit Clifford operations.

\section{Conclusion}\label{Sec:Conclude}

The fact that, in any finite dimension, the Clifford group is generated by generalizations of the qubit Hadamard, Phase-shift, and CNOT gates shows that the group generalizes more naturally than previously expected. Furthermore, the algorithmic approach used to prove this result provides a straightforward and efficient means of implementing arbitrary qudit Clifford operations in any finite dimension using a minimal set of distinct gates. Thus, along with this inductive decomposition of Clifford operators, we can say that the QFT, Phase-shift, and SUM gates serve as a minimal, universal, and constructable characterization of the Clifford operators. Under the assumption that these gates can be built, this characterization is also physical, serving as an ideal characterization of the Clifford group.

Furthermore, because this characterization is closed, we need not necessarily assume external resources, such as the ability to make projective measurements or the existence of ancillas, when computing Clifford circuit complexity. That said, under certain physical constructions, it may be preferrable and/or easier to use external resources when implementing these Clifford operations. Indeed, in \cite{GKP_PRA2001} is was shown that the Clifford basis gates could be physically realized via linear optics and ``sqeezing,'' while other constructions that implement Clifford operations as a subset (e.g. \cite{BdeGS_PRA02}, \cite{Strauch_PRA11} and references therein) require that the quantum system be coupled with an external ancillary system. Nevertheless, this closed characterization exists apart from any single physical construction, and is hence not limited to any particular implementation.

Applying these results to the GKP embedding procedure, we find that for any symmetric embedding of a qunit into a qudit, the qudit Clifford basis gates in fact implement their respective logical gates, and hence act as a logical basis for the logical Clifford group. Though this is not generally true for asymetric embeddings, in general we need only check two logical transformations to determine whether arbitrary logical qunit Clifford operations exist within the actual qudit Clifford framework. Further research is warrented to determine exactly what conditions are required on an embedding so that the qunit Clifford gates can always be characterized within the qudit Clifford framework.


Lastly, we point out that this paper characterizes the $n$-qudit Clifford group in which all $n$ quantum systems exist in Hilbert spaces of the same dimension. This characterization does not apply to operators that act on a collection of systems of differing dimensions, that is, ones that exists in the space $\mc{H} = \mc{H}_{d_1} \otimes \mc{H}_{d_2} \otimes \cdots \otimes \mc{H}_{d_n}$, where each $\mc{H}_{d_i}$ is a $d_i$-dimensional complex Hilbert space. An open question is whether a classical representation exists that allows for the unique characterization of the corresponding Clifford operators acting on this ``mixed'' state space, up to global phase, via a minimal gate set.

\ack

The author would like to thank J.~E.~Troupe and A.~D.~Parks, for their helpful review and discussions. Additionally, the author would like to thank M.~Grassl for helpful comments on a previous version of this manuscript, and for providing additional resources. This research was funded by the Naval Surface Warfare Center, Dahlgren Division (NSWCDD) In-house Laboratory Independent Research (ILIR) program.

\appendix
\section{PEG Algorithm Applied to Single-Qudit Symplectic Matrices}\label{PEG-Algorithm}
The Euclidean factoring algorithm is a method of determining the greatest common divisor of two nonzero integers, and takes $\mc{O}(\ln(D))$ steps when working over $\mbb{Z}_D$. It works in the following way. If we would like to compute $\gcd(a,b)$ for $a,b \in \{1, 2, \dots, D-1\}$, then suppose $a = m_1b + c_1 \pmod{D}$. It follows that $\gcd(a,b) = \gcd(b,c_1)$. Likewise, we can then write $b = m_2c_1 + c_2 \pmod{D}$, and hence $\gcd(a,b) = \gcd(b,c_1) = \gcd(c_1, c_2)$. The Euclidean factoring algorithm works by iterating the above steps until we obtain a $c_k$ such that $c_k = 0$. Then $c_{k-1} = \gcd(a,b)$.

Let $M = \begin{bmatrix} p&q\\r&s \end{bmatrix}$ be symplectic, and suppose none of $p, q, r, s$ are invertible. Since $M$ is invertible, then for any column $M_{.,i}$ of $M$, it follows that $\gcd(M_{.,i})$ is invertible; in particular, $\gcd(q,s)$ is invertible. The goal of the PEG Algorithm is to create an $M^\p = \begin{bmatrix}p^\p & q^\p\\ r^\p & s^\p \end{bmatrix}$ such that $q^\p = \gcd(q,s)$. Following the Euclidean algorithm, suppose $q = m_1s + c_1 \pmod{D}$. Then we apply $(RP^{m_1}R^3)$ to $M$ to get 
\begin{equation}
(RP^{m_1}R^3)M = \begin{bmatrix} p-m_1r & c_1\\ r & s\end{bmatrix} = M^\p_1.
\end{equation}
Suppose now that $s = m_2c_1 + c_2 \pmod{D}$. Then we apply $P^{-m_2}$ to $M_1^\p$ to get 
\begin{equation}
P^{-m_2}M_1^\p = \begin{bmatrix}p-m_1r & c_1\\ r-m_2(p-m_1r) & c_2\end{bmatrix} = M_2^\p.
\end{equation}
Continue in this fashion to end the Euclidean algorithm. In the end, we want a matrix of the form
\begin{equation}
M_{Fin}^\p = \begin{bmatrix} p^\p & \gcd(q,s)\\ r^\p & 0 \end{bmatrix}.
\end{equation}
If, at the end of the algorithm, we instead have a matrix of the form 
\begin{equation}
M_k^\p = \begin{bmatrix}p^\p & 0\\ r^\p & \gcd(q,s)\end{bmatrix},
\end{equation}
then add one more step, namely, apply $R^3$ to $M_k^\p$ to get
\begin{equation}
R^3M_k^\p = \begin{bmatrix} r^\p & \gcd(q,s)\\ -p^\p & 0 \end{bmatrix} = M_{Fin}^\p.
\end{equation}

Since $\gcd(q,s)$ is invertible, we now have a matrix in the form described in Case \ref{Case:q-invert} in Section \ref{Sec:BuildClifford}. Thus, we can build $M_{Fin}^\p$ straightforwardly. We obtain our original matrix $M$ by sequentially applying the inverses of the operations used in the PEG Algorithm.

As a very simple example, suppose
\begin{equation}
M = \begin{bmatrix} 10&9\\3&4\end{bmatrix},
\end{equation}
where the entries are over $\mbb{Z}_{12}$. $M$ is symplectic, since $\det(M) = 13 \equiv 1 \pmod{12}$. None of the entries of $M$ are invertible, since none are relatively prime to 12. Then following the PEG algorithm, we observe that $9 = 2\cdot 4 + 1$. So we apply $(RP^2R^3)$ to $M$ to get
\begin{equation}
M^\p_1 = (RP^2R^3)M = \begin{bmatrix}1&-2\\0&1\end{bmatrix}M = \begin{bmatrix}4&1\\3&4\end{bmatrix}.
\end{equation}
Following the algorithm to completion, we observe that $4 = 4\cdot 1 + 0$, and so we apply $P^{-4} = P^8$ to $M_1^\p$ to get
\begin{equation}
M_2^\p = P^8M_1^\p = \begin{bmatrix}1&0\\8&1\end{bmatrix}M_1^\p = \begin{bmatrix}4&1\\11&0\end{bmatrix}.
\end{equation}
Now $M_2^\p$ has the form described in Case \ref{Case:q-invert}, and so we know that $M_2^\p = PRPRP^5$. To obtain the original $M$, we simply apply $P^{-8} = P^4$ on the left, followed by $(RP^2R^3)^{-1} = RP^{10}R^3$. Thus, the $M$ in our example is decomposed into $M = RP^{10}R^3P^5RPRP^5$.

\section{PEG Algorithm Adapted for SUM Gate}\label{Appendix:PEG-SUM}
We show here how to implement an algorithm similar to the PEG Algorithm that uses the SUM gate to map the vectors $(0,0,a,b)^T$ to either $(0,0, \gcd(a,b), 0)^T$ or $(0,0,0, \gcd(a,b))^T$ for $a, b \in \{1, 2, \dots, d-1\}$. Just as in the PEG Algorithm, suppose $a = m_1b + c_1 \pmod{d}$, so that $\gcd(a,b) = \gcd(b, c_1)$. It follows that
\begin{equation}
C_{[1,2]}^{m_1}(0,0,a,b)^T = \begin{bmatrix} 1&0&0&0\\m_1&1&0&0\\0&0&1&-m_1\\0&0&0&1 \end{bmatrix} \begin{bmatrix} 0\\0\\a\\b \end{bmatrix} = \begin{bmatrix} 0\\0\\c_1\\b \end{bmatrix}.
\end{equation}
Now suppose that $b = m_2c_1 + c_2$ so that $\gcd(c_1,b) = \gcd(c_1, c_2)$. In this case, we switch the control and target qudit, described by applying a product of $C_{[2,1]} = C_{[1,2]}^T$:
\begin{equation}
C_{[2,1]}^{m_2}(0,0,c_1,b)^T = \begin{bmatrix} 1&m_2&0&0\\0&1&0&0\\0&0&1&0\\0&0&-m_2&1 \end{bmatrix} \begin{bmatrix} 0\\0\\c_1\\b \end{bmatrix} = \begin{bmatrix} 0\\0\\c_1\\c_2 \end{bmatrix}.
\end{equation}
Following in this fashion, the algorithm converges to either of the vectors $(0,0,\gcd(a,b),0)^T$ or $(0,0,0,\gcd(a,b))^T$. Suppose, in this case, that the final step of the algorithm returns the vector $(0,0,0,\gcd(a,b))^T$, but the vector $(0,0,\gcd(a,b),0)^T$ is desired. This is easily remedied by observing that $C_{[2,1]}C_{[1,2]}^{D-1}(0,0,0,\gcd(a,b))^T = (0,0,\gcd(a,b),0)^T$. It follows that, in the unitary case, up to some global phase, using sequences of the $\ov{C}$ gate with choice of control and target qudit, we can map any $Z^a\otimes Z^b$ to either $I \otimes Z^{\gcd(a,b)}$ or $Z^{\gcd(a,b)} \otimes I$ for any $a, b \in \{1, 2, \dots, d-1\}$. Note that we could have alternatively applied the SWAP operator $S_{[1,2]}$ in this last step, but chose to use the above method instead to show that the entire algorithm could be performed using only SUM operators.

\section*{References}
\bibliographystyle{unsrt}
\bibliography{bibfile}

\end{document}